\documentclass[a4paper,11pt, notitlepage]{article}
\usepackage[utf8]{inputenc}
\usepackage[margin=1in]{geometry}
\usepackage{graphicx}
\usepackage{amsfonts}
\usepackage{tikz}
\usepackage{hyperref}
\usepackage{multirow}
\usepackage{amsmath, amsthm}
\usepackage[linesnumbered, algoruled]{algorithm2e}
\usepackage{microtype}

 \newtheorem{lemma}{Lemma}
 \newtheorem{theorem}{Theorem}
\newtheorem{fact}{Fact}
 \newtheorem{corollary}{Corollary}

\newcommand{\wadapt}{\textit{Weakly-Adaptive}}
\newcommand{\noadapt}{\textit{Non-Adaptive}}
\newcommand{\chainadapt}{\textit{$k$-Chain-Ordered}}
\newcommand{\kthick}{\textit{$k$-Ordered}}
\newcommand{\sadapt}{\textit{Strongly-Adaptive}}
\newcommand{\ordered}{\textit{Ordered}}
\newcommand{\constrained}{\textit{Constrained}}

\newcommand{\alga}{\textsc{Algorithm A}}
\newcommand{\algb}{\textsc{Algorithm B}}
\newcommand{\algc}{\textsc{Algorithm C}}

\newcommand{\elect}{\textsc{Elect-Leader}}
\newcommand{\electc}{\textsc{Gather-Leaders}}
\newcommand{\countstations}{\textsc{Count-Processes}}
\newcommand{\lcon}{\textsc{Leader-Consensus}}
\newcommand{\prop}{\textsc{Propagate-Msg}}
\newcommand{\synran}{\textsc{SynRan}}

\newcommand{\timelc}{T_{LC}}

\newcommand{\polylog}{{\mbox{polylog}}}

\newcommand{\remove}[1]{}

\newcommand{\dk}[1]{{\color{black} #1}}
\newcommand{\jm}[1]{{\color{black} #1}}

\newcommand{\cO}{{\mathcal O}}

\begin{document}

\title{On the complexity of fault-tolerant consensus
}

\author{
Dariusz R. Kowalski\footnote{Department of Computer Science,
	University of Liverpool, Ashton Building, Ashton Street,
	Liverpool L69 3BX, UK. 
	Email: \texttt{\{d.kowalski,j.mirek\}@liverpool.ac.uk}.} \and
Jaros\l{}aw Mirek\footnotemark[1]}

\date{}

\maketitle

\begin{abstract}
The paper studies the problem of reaching agreement in a distributed message-passing system prone to crash failures. 
Crashes are generated by \constrained\ adversaries - a \wadapt\ adversary, who has to fix in advance the set of $f$ crash-prone processes, 
or a \chainadapt\ adversary, who orders all the processes into $k$ disjoint chains and has to follow this pattern when crashing them. 
Apart from these constraints, both of them may crash processes in an adaptive way at any time.
While commonly used \sadapt\ adversaries model attacks
and \noadapt\ ones -- pre-defined faults, the constrained
adversaries model more realistic scenarios when there are
fault-prone dependent processes, e.g., in hierarchical 
or dependable
software/hardware systems.
We propose time-efficient consensus algorithms against such adversaries and also show how to improve the message complexity of proposed solutions.
Finally, we show how to reach consensus against a \kthick\ adversary, limited by an arbitrary partial order \dk{with a maximal anti-chain of length $k$}.
We complement our algorithmic results with (almost) tight lower bounds, and extend the one for \wadapt\ adversaries
to hold also for (syntactically) weaker \noadapt\ adversaries. Together with the consensus algorithm against \wadapt\ adversaries (which automatically translates to \noadapt\ adversaries), these results extend  the state-of-the-art of the popular class of \noadapt\ adversaries, 
in particular the result of Chor, Meritt and Shmoys~\cite{CMS},
and prove general separation between \sadapt\ and the constrained adversaries (including \noadapt) 
analyzed by Bar-Joseph and Ben-Or~\cite{BB} and others.
\end{abstract}

\noindent

\vspace*{-2ex}
\section{Introduction}

We study the problem of consensus in synchronous message passing distributed systems.
There are $n$ processes, out of which at most $f$ can crash.
Each process is initialized with a binary input value, and
the goal is to agree on one value from the input ones by all processes. 
Formally, the following three properties need to be satisfied:
{\em agreement:} no two processes decide on different values; 
{\em validity:} only a value among the initial ones may be decided upon; 
and
{\em termination:}	each process eventually decides, unless it crashes.
In case of randomized solutions, the specification of consensus needs to be reformulated, which can be done in various ways (c.f.,~\cite{AtWe}).
We consider a classic reformulation in which validity and agreement are required to hold for every execution, while termination needs to hold with probability~$1$.

Efficiency of algorithms is measured twofold: by the the number of rounds (time complexity) and by the total
number of point-to-point messages sent (message complexity) until all non-faulty processes decide.
This work focuses on \emph{efficient randomized solutions} -- time and communication performance metrics are understood in expected sense \dk{(though we also provide some analysis on
	work performance per process and time complexity formulas depending on the desired probability of achieving it)}.
%
%
%
Randomization has been used in consensus algorithms for various kinds of failures specified by adversarial models, see~\cite{Asp-DC,AtWe}.
Reason for considering randomization is to overcome inherent limitations of deterministic solutions.
Most surprising benefits of randomization is the solvability of consensus in as small as constant time~\cite{CMS,FM,Rab}.
Feasibility of achieving small upper bounds on performance of algorithms solving consensus in a given distributed environment depends on the power of adversaries inflicting failures.


\dk{On the other hand,} Chlebus, Kowalski and Strojnowski~\cite{CKS} observed that consensus cannot be solved \emph{deterministically} by an algorithm that is locally scalable (i.e., slowly growing performance function per each process) with respect to message complexity. 
They developed a globally scalable (i.e., slowly growing performance
function per average process) deterministic solution with respect to bit communication that can cope with any number $f<n$ of crash failures.
Their solution is not explicit in the sense that it relies on families of graphs with properties related to graph expansion that are shown to exist by the probabilistic method.
\dk{In this work, we investigate scalability of randomized algorithms against \constrained\
adversaries.}





\subsection{Previous and related work}
\label{s:previous}

\dk{Consensus is one of the fundamental problems in distributed computing,
with a rich history of research done in various settings
and systems, c.f.,~\cite{AtWe}. Recently its popularity grew even further due to
applications in emerging technologies such as blockchains.
Below we present only a small digest of literature
closely related with the setting considered in this work.}


Consensus is solvable in synchronous systems with processes prone to crashing, although time $f+1$ is required~\cite{FL} and sufficient~\cite{GM93} in case of deterministic solutions. 
Chor, Meritt and Shmoys~\cite{CMS} showed that randomization allows to obtain a constant expected time algorithm against 
a \noadapt\
adversary, if the minority of processes may crash.
Bar-Joseph and Ben-Or~\cite{BB} proved tight upper and lower bounds $\Theta(f/\sqrt{n\log n})$ on the expected time of randomized consensus solutions for the \sadapt\ adversary 
that can decide to fail $f=\Omega(n)$ processes.


Message complexity of consensus was extensively studied before.
Amdur, Weber and Hadzilacos~\cite{AWH} showed that $\Omega(n)$ messages need to be sent by a protocol solving consensus in an eventually synchronous environment in a failure-free execution.
Dolev and Reischuk~\cite{DR} and Hadzilacos and Halpern~\cite{HH} proved the $\Omega(fn)$ lower bound on the message complexity of deterministic consensus for Byzantine failures.

Chlebus and Kowalski~\cite{CK-SPAA-09} showed that consensus can be solved by a randomized algorithm that is locally scalable with respect to both time and bit communication complexities, when the number $f$ of failures is at most a constant fraction of the number $n$ of processes and the adversary is \noadapt. 
A  comparable performance is impossible to achieve by a deterministic solution, as was observed in~\cite{CKS}.
A deterministic algorithm solving consensus in the synchronous setting that has processes sending a total of $\cO(n\, \polylog n)$ bits per message, meaning the algorithm is globally scalable, and which can handle up to $n-1$ crashes, has been developed in~\cite{CKS}.
The globally scalable solution given in~\cite{CKS} is deterministic but it is not explicit.

Gilbert and Kowalski~\cite{GK-SODA-10} presented a randomized consensus algorithm that achieves optimal communication complexity, using $ \mathcal{O}(n) $ bits of communication and terminates in 
$ \mathcal{O}(\log n) $ time with high probability, tolerating up to $ f < n/2 $ crash failures against a \noadapt\ adversary.




\subsection{Our results}

\begin{table*}[htbp]
\begin{center}
\scalebox{0.7}{
  \begin{tabular}{|c|c|c|c|c|}
    \hline
      &  & \textbf{\sadapt} & \textbf{\wadapt} and \textbf{\noadapt} & \textbf{\chainadapt} \\\hline
      
    \multirow{2}{*}{\textbf{randomized}} & \small{upper bound} & $ \mathcal{O}\left(\sqrt{\frac{n}{\log n}}\right)$ \cite{BB} & $ \mathcal{O}\left(\sqrt{\frac{n}{(n-f)\log(n/(n-f))}}\right) $ * & $ \mathcal{O}\left(\sqrt{\frac{k}{\log k}}\log(n/k)\right)$ *\\ & \small{lower bound} & $ \Omega\left(\sqrt{\frac{n}{\log n}}\right)$ \cite{BB} & $ \Omega\left(\sqrt{\frac{n}{(n-f)\log(n/(n-f))}}\right) $ * & $\Omega\left(\sqrt{\frac{k}{\log k}}\right)$ *\\\hline
    \multirow{2}{*}{\textbf{deterministic}} & \small{upper bound} &  \multicolumn{3}{c|}{$ f + 1 $ \cite{GM93}} \\ & \small{lower bound} &\multicolumn{3}{c|}{$ f + 1 $ \cite{FL}} \\\hline
    
  \end{tabular}
}
\end{center}
  \caption{Time complexity of solutions for the consensus problem against different adversaries. Formulas with * are presented in this paper.}
  \label{tab1}
  \end{table*}

\sloppy 
\dk{We analyze the consensus problem against restricted adaptive
adversaries. The motivation is that 
a \sadapt\
adversary, typically used for analysis of randomized consensus algorithms,
may not be very realistic; for instance, in practice some
processes could be set as fault-prone in advance, before
the execution of an algorithm, or may be dependent i.e.,
in hierarchical hardware/software systems.
In this context, a \sadapt\ adversary should be used
to model attacks rather than realistic crash-prone systems.
On the other hand, a 
\noadapt\ 
adversary who must fix
all its actions before the execution does not capture
many aspects of fault-prone systems, e.g., attacks or
reactive failures (occurring as an unplanned consequence of 
some actions of the algorithm in the system).
Therefore, analyzing the complexity of consensus under such
constraints gives a much better estimate on what may happen
in real executions and, as we demonstrate, leads to new
interesting theoretical findings about the performance of 
consensus algorithms.}

Table \ref{tab1} presents time complexities of solutions for the consensus problem against different adversaries.
\remove{
For deterministic algorithms, \constrained\ adversaries are consistent with the \sadapt\ one, 
as they all could simulate the algorithm before its execution and determine best decisions.
}
\dk{Results for the \sadapt\ adversary and for deterministic
algorithms are known (see Section~\ref{s:previous}), while the other ones are delivered
in this work.}
We design and analyze a randomized algorithm that reaches consensus 
in expected $ \mathcal{O}\left(\sqrt{\frac{n}{(n-f)\log(n/(n-f))}}\right) $ rounds
using 
$\mathcal{O}\left(n^2+\sqrt{\frac{n^5}{(n-f)^5\log(n/(n-f))}}\right)$
point-to-point messages, in expectation,
against any \wadapt\ adversary that may crash up to 
$f<n$ processes.
This result is time optimal due to the proved lower bound
$\Omega\left(\sqrt{\frac{n}{(n-f)\log(n/(n-f))}}\right)$
on expected number of rounds.

The lower bound could be also generalized to hold against
the (syntactically) weaker \noadapt\ adversaries, therefore all the results concerning \wadapt\ adversaries delivered in this paper hold for \noadapt\ adversaries as well.
This extends the state-of-the-art of the study of \noadapt\ adversaries done in high volume of previous work,
c.f.,~\cite{CMS,CK-SPAA-09,GK-SODA-10},
specifically, an $O(1)$ 
expected time algorithm of Chor et al.~\cite{CMS} only for a constant (smaller than $1$) fraction of failures. 
Our lower bound is the first non-constant formula 
depending on the number of
crashes proved for this adversary. 
%
In view of the lower bound $\Omega\left(\frac{f}{\sqrt{n\log n}}\right)$ \cite{BB} on the expected number of
rounds of any consensus algorithm against a \sadapt\
adversary crashing at most $f$ processes, 
\emph{our result
	shows separation of the two important classes of adversaries -- 
	\noadapt~and \sadapt\ -- 
for one of the most fundamental problems in distributed computing, 
which~is~consensus.}

\dk{Furthermore we show an extension of the above mentioned algorithm} accomplishing consensus in
slightly (polylogarithmically) longer number $\mathcal{O}\left(\sqrt{\frac{n\log^5 n}{(n-f)}}\right)$ of expected rounds but
using a substantially smaller number of $\mathcal{O}\left(\left(\frac{n}{n-f}\right)^{3/2} \log^{7/2} n + n\log^4 n\right)$ point-to-point messages in expectation.
Time and message complexities of the latter algorithm are 
only polylogarithmically far from the lower bounds for almost all values of $f$.
More specifically, only for $f=n-\Omega(\sqrt{n})$ the message complexity is away from the lower bound by some small polynomial.

We complement the results by showing how to modify the algorithm designed for the \wadapt\ adversary, to work against a \chainadapt\ adversary, who has to arrange all processes into an order of $ k $ chains, and then has to preserve this order of crashes in the course of the execution.
The algorithm reaches consensus in $ \mathcal{O}\left(\sqrt{\frac{k}{\log k}}\log(n/k)\right) $
\dk{rounds in expectation}. \jm{Additionally, we show a lower bound $\Omega\left(\sqrt{\frac{k}{\log k}}\right)$ for the problem against a \kthick\ adversary.}
Finally, we show that this solution is capable of running against an arbitrary partial order with a maximal anti-chain of size $ k $.
\dk{Similarly to results for the \wadapt\ adversary, formulas obtained for \ordered\ adversaries separate
them from \sadapt\ ones.}

For algorithms proposed in this paper, we also analyze expected work per process and
give time complexity formulas depending on 
the desired probability of achieving it.

Due to space limitations some parts of this work were deferred to the Appendix.
%

\section{Model}

\paragraph{\textbf{Synchronous distributed system.}}
We assume having a system of $ n $ processes that communicate in the message passing model. This means that processes form a complete graph where each edge represents
a communication link between two processes. If process $ v $ wants to send a message to process $ w $, then this message is sent via link $ (v,w) $. It is worth noticing that links are symmetric, i.e.,
$ (v,w) = (w,v)$. We assume that messages are sent instantly.

Following the synchronous model by \cite{BB}, we assume that computations are held in a synchronous manner and hence time is divided into rounds consisting of two phases:

\noindent
 -- Phase A - generating local coins and local computation.
 
 \noindent
 -- Phase B - sending and receiving messages.

\paragraph{\textbf{Adversarial scenarios.}}

Processes are prone to crash-failures that are a result of the adversary activity. The adversary of our particular interest is an adaptive one - it can make arbitrary decisions and see all local computations
and local coins, as well as messages intended to be sent by active processes. Therefore, it can decide to crash processes during phase B. Additionally while deciding that a certain process will crash, it can decide which subset of messages will reach their recipients.

\noindent
In the context of the adversaries in this paper we distinguish three types of processes:

\noindent
 -- Crash-prone - processes that can be crashed by the adversary.
 
 \noindent
 -- Fault-resistant - processes that are not in the subset of the \wadapt\ adversary and hence cannot be crashed.
 
 \noindent
 -- Non-faulty - processes that survived until the end of the algorithm.

\noindent
\textbullet \; \textit{\sadapt\ and \wadapt\ adversaries.}
 The only restriction for the \sadapt\ adversary is that it can fail up to $ f $ processes, where $ 0 \leq f < n $.
 
 The \wadapt\ adversary is restricted by the fact that before the algorithm execution it must choose $ f $ processes that will be prone to crashes,
where $ 0 \leq f < n $. 

Observe that for deterministic algorithms the \wadapt\ adversary is consistent with the \sadapt\ adversary, because it could simulate the algorithm before its execution
and decide on choosing the most convenient subset of processes.

\noindent
\textbullet \; \textit{\chainadapt\ and \kthick\ adversaries.}
The notion of a \chainadapt\ adversary originates from partial order relations, hence appropriate notions and definitions translate straightforwardly.
The relation of our particular interest while considering partially ordered adversaries is the precedence relation.
Precisely, if some process $ v $ precedes process $ w $ in the partial order of the adversary, then we say that $ v $ and $ w $ are comparable. This means that
process $ v $ must be crashed by the adversary before process $ w $. 
Consequently a subset of processes where every pair of processes is comparable is called a chain. On the other hand a subset of processes where no two different processes are comparable is called an anti-chain.

It is convenient to think about the partial order of the adversary from a Hasse diagram perspective. The notion of chains and anti-chains seems to be intuitive when graphically presented, e.g., a chain
is a pattern of consecutive crashes that may occur while an anti-chain gives the adversary freedom to crash in any order due to non-comparability of processes.

Formally, the \chainadapt\ adversary has to arrange \textbf{all} the processes into a partial order consisting of $ k $ disjoint chains of arbitrary length that represent in what order these processes may be crashed. 

By the \textit{thickness} of a partial order $P$ we understand the maximal size of an anti-chain in~$P$.
An adversary restricted by a wider class of partial orders of thickness $k$
is called a \kthick\ adversary.

We refer to a wider class of adversaries in this paper, constrained by an arbitrary partial order, as \ordered\ adversaries. What is more, adversaries having additional limitations, apart from the possible number of crashes (i.e. all described in this paper but the \sadapt\ adversary), will be called \constrained\ adversaries.

Note that \ordered\ adversaries are also restricted by the number of possible crashes $ f $ they may enforce. Additionally, throughout this paper we denote $ l_{j} $ as the length of chain $ j $

\noindent
\textbullet \; \textit{\noadapt\ adversaries.}
The \noadapt\ adversaries are characterised by the fact that they must fix all their decisions prior to the execution of the algorithm and then follow this pattern during the execution.
%

\paragraph{\textbf{Consensus problem.}}
In the consensus problem $ n $ processes, each having its input bit $ x_{i} \in \{0, 1\}$, $ i \in \{1,\dots,n\}$, have to agree on a common output bit in the presence of the adversary, capable of crashing processes.
We require any consensus protocol to fulfill the following conditions:


\noindent
 -- Agreement: all non-faulty processes decide the same value. 
 
 \noindent
 -- Validity: if all processes have the same initial value $ v $, then $ v $ is the only possible decision value.
 
 \noindent
 -- Termination: all non-faulty processes decide with probability~$1$.

\dk{We follow typical assumption that the first two requirements
must hold in \emph{any} execution, while termination should be satisfied with probability $1$.}

\paragraph{\textbf{Complexity measures.}}
The main complexity measure used to benchmark the consensus problem is the number of rounds by which all non-faulty processes decide on a common value.

On the other hand we discuss on the message complexity, understood as the total number of messages sent in the system, until all non-faulty processes decide on a common value.

\jm{The expected work per process is the expected number of non-idle computational steps of a process accrued during the execution of an algorithm.}

\paragraph{\textbf{Algorithmic tools.}} \label{tools}
Throughout this paper we use black-box fashioned procedures that allow us to clarify the presentation. We now briefly describe their properties and refer to them in the algorithms' analysis.

\noindent
\textit{\lcon\ properties.} Because we treat the \lcon\ procedure as a black-box tool for reaching consensus on a small group of processes, we need to define properties that we expect to be met by
this protocol:

\noindent
 \textbullet \; 
 \lcon\ reaches consensus on a group of no more than $ \frac{6n}{n-f} $ processes in time $\timelc(\frac{6n}{n-f})$ with probability at least $ \frac{9}{10} $, where 
 $ \timelc $ is time complexity of \lcon.
 
 \noindent
 \textbullet \; \lcon\ satisfies termination, validity and agreement.

A candidate solution to serve as \lcon\ is the Ben-Or and Bar-Joseph's \synran\ algorithm from \cite{BB} and we refer the reader to the details therein.

\noindent
\textit{\prop\ properties.} We assume that procedure \prop\ propagates messages in $ 1 $ round with $ \mathcal{O}(n^{2}) $ message complexity. This is consistent with a scenario where
full communication takes place and each process sends a message to all the processes.
 
\section{\wadapt\ adversary}

In this section we consider the fundamental result i.e. \alga\ that consists of two main components - a leader election procedure, 
and a reliable consensus protocol. We combine them together in an appropriate way, in order to reach consensus against a \wadapt\ adversary. 

\subsection{\alga}

\begin{algorithm}
{initialize list $\texttt{LEADERS} $ to an empty list\;}
{$decided := false$\;}
{$value := x_{v}$\;}
{\Repeat{$decided$}
    {
    {$ \texttt{LEADERS} := \elect $\;}

    {\If{$\texttt{LEADERS}$ contains $v$}{\lcon($\texttt{LEADERS}$) \;
      {\If{consensus achieved}{execute \prop($value$) twice\;
      $ decided = true $\;}
      \Else{wait for two rounds\;}}}  
    }
    {\Else{idle for $\timelc$ rounds\;
	   {\If{heard the same consensus value $CV_{w}$ twice from some process $w$}{$value := CV_{w}$\;
	      $decided = true $\;		     	
	   }
	   {\If{heard consensus value $CV_{w}$ once from some process $w$}{$value := CV_{w}$\;}
	   }
   }
   }
   }
   {clear list $ \texttt{LEADERS} $\;}
}
}

\caption{\alga, \dk{pseudocode} for process~$ v $ }
\label{algorithm1}
\end{algorithm}

\begin{algorithm}
{$\texttt{coin} := \frac{1}{n-f}$\;}
{initialize list $ \texttt{LEADERS} $ to an empty list\;}
{
{toss a coin with the probability $ \texttt{coin} $ of heads to come up\;}
{\If{heads came up in the previous step}{
    \prop(``$ v $'') to all other processes\;
    add $ v $ to list $\texttt{LEADERS}$\;}
}
fill in list $ \texttt{LEADERS}$ with elected leaders' identifiers from received messages\;
}
return $\texttt{LEADERS}$\;
\caption{\elect, \dk{pseudocode} for process~$v$}
\label{algorithm2}
\end{algorithm}


\alga\ has an iterative character and begins with a leader election procedure in which we expect to elect $ \mathcal{O}(\frac{n}{n-f}) $ leaders simultaneously. The leaders run \lcon\ procedure in which they are expected to reach consensus within
their own group. We call this a \textit{small consensus}. If they do so this fact is propagated to all the processes via \prop\ so that they know about reaching the small consensus and set their consensus values accordingly. Communicating this fact, implies reaching \textit{consensus} by the whole system.

Let us follow Algorithm \ref{algorithm1} from the perspective of some process $ v $. At the beginning of the protocol every process takes part in \elect\ procedure and process $ v $ tosses
a coin with probability of success equal $ \frac{1}{n-f} $ and either is chosen to the group of leaders or not. If it is successful, then it communicates this fact to all processes.

Process $ v $ takes part in \lcon\ together with other leaders in order to reach a small consensus and if the consensus value is fixed $ v $ tries to convince other processes to this value twice. This is because if some process $ w \neq v $ receives the consensus value in the latter round, then it may be sure that other processes received this value from $ v $ as well in the former round
(so in fact every process has the same consensus value fixed). Process $ v $ could not propagate its value for the second time if it was not successful in propagating this value to every other process for the first time - if just one process did not receive this value this would be consistent with a crash of $ v $.

If process $ v $ was not chosen a leader then it listens the channel for an appropriate amount of time and afterwards tries to learn the consensus value twice. If it is 
unable to hear a double-consensus, then it idles for two rounds for the purpose of synchronization. If consensus is not reached, then the protocol starts again with electing another group of leaders. Nevertheless, if process $ v $ hears a consensus value once, it holds and assigns it as a candidate consensus value. This guarantees continuity of the protocol. 

The idea standing behind \alga\ is built on the fact that if just one fault-resistant process is elected to the group of leaders then the adversary is unable to crash it in the course of an execution.
If additionally, the small consensus protocol works well and all messages are propagated properly, then consensus is achieved after a certain number of rounds.

\subsection{Analysis of \alga}

We begin with a lemma stating that \elect\ is likely to elect $ \mathcal{O}(\frac{n}{n-f}) $ leaders. 

\begin{lemma}
 Let $ L $ denote the number of leaders elected by \elect. $ L < 6\frac{n}{n-f} $ at least with probability $ \frac{9}{10} $.
 \label{lemma1}
\end{lemma}

\begin{lemma}
 Procedure \elect\ elects at least one sustainable, fault-resistant leader with probability at least $ \frac{6}{10} $.
 \label{lemma2}
\end{lemma}

\begin{theorem} \label{theorem1}
The following hold for \alga\ against the \wadapt\ adversary:

  \noindent
  \textbf{a)} \alga\ reaches consensus in the expected number of rounds equal $ \mathcal{O}\left(\timelc\left(\frac{n}{n-f}\right)\right)$, satisfying termination, agreement and validity.
  
  \noindent
  \textbf{b)} Expected work per process of \alga\ is $ \frac{1}{n-f} \; \mathcal{O}\left(\timelc\left(\frac{n}{n-f}\right)\right) $.
  
  \noindent
  \textbf{c)} For some fixed $ \epsilon > 0 $, \alga\ reaches consensus in 
  $ \mathcal{O}\left(\log\left(\frac{1}{\epsilon}\right) \; \timelc\left(\frac{n}{n-f}\right)\right) $ rounds and with work per process $ \frac{1}{n-f} \; \mathcal{O}\left(\log\left(\frac{1}{\epsilon}\right) \;\timelc\left(\frac{n}{n-f}\right)\right) $ with probability $ 1 - \epsilon $.

\end{theorem}

\begin{proof}
   
   Let us describe three events and their corresponding probabilities:
   
    \noindent
    - (Lemma \ref{lemma1}) \elect\ chooses less than $ 6\frac{n}{n-f} $ leaders, with prob. bigger than~$\frac{9}{10}$.
    
    \noindent
    - (Lemma \ref{lemma2}) \elect\ chooses at least one fault-resistant leader, with prob. greater than~$\frac{6}{10} $.
    
    \noindent
    - (c.f. Section \ref{tools}) \lcon\ finishes the small consensus within its expected time $\timelc(\frac{6n}{n-f})$ , with probability~$\frac{9}{10}$. 
   
   Let us define the conjunction of these events as the \textit{success of the algorithm} and denote it as $ Z $.
   Since for any events $ A $ and $ B $ we have that $ \mathbb{P}(A \cap B) \geq \mathbb{P}(A) + \mathbb{P}(B) - 1 $,
   it is easy to observe that $ Z $ happens with probability at least $ \frac{2}{5} $.
   
   \noindent
\textbf{a)}
    If $ Z $ holds then \alga\ reaches consensus in time $ \mathcal{O}(\timelc(\frac{6n}{n-f}))$.   
   In what follows since $ Z $ takes place with probability $ \frac{2}{5} $, then  $ \frac{5}{2} $ is the expected number of iterations of \alga\ until reaching consensus. Consequently \alga\ reaches consensus in the expected number of rounds equal
   $ \mathcal{O}(\timelc(\frac{6n}{n-f})) $.
   
   Because the termination condition is satisfied by \lcon\ and the expected number of iterations of \alga\ is a constant, \alga\ also meets the
   termination condition.

   The validity condition is also satisfied, by the validity of the \lcon\ protocol and the fact that 
   \alga\ does not change values held by processes. 
   
   In order to prove the agreement condition let us suppose to the contrary that in some execution $\mathcal{A}$ there is process $ v $ that reached consensus with value $ 1 $ 
   and process $ w $ that reached consensus with value $ 0 $. Process $ v $ did set value $ 1 $ as a result of the small consensus reached by the \lcon\ procedure. 
   So did process $ w $ with value $ 0 $. This contradicts the agreement condition of \lcon.

\end{proof}

\begin{corollary}
 Instantiating \lcon\ with \synran\ from \cite{BB} results in:

  \noindent
  \textbf{a)} $ \mathcal{O}\left(\sqrt{\frac{n}{(n-f)\log(n/(n-f))}}\right) $ expected rounds to reach consensus by \alga,
  
  \noindent
  \textbf{b)} expected work per process equal $ \mathcal{O}\left(\sqrt{\frac{n}{(n-f)^3\log(n/(n-f))}}\right) $,
  
  \noindent
  \textbf{c)} $ \mathcal{O}\left(\log\left(\frac{1}{\epsilon}\right)\sqrt{\frac{n}{(n-f)\log(n/(n-f))}}\right) $ rounds to reach consensus by \alga\ with probability $ 1 - \epsilon $, for $ \epsilon > 0 $.

\end{corollary}

\alga\ is designed in the way that most of the processors wait until the consensus value is decided and propagated through the communication medium.
Expected work per process value leads to an observation that, considering more practical scenarios, most of the processors can take care about different processes, while only a subset is responsible for actually reaching consensus.
 
\subsection{Improving message complexity}


The expected number of point-to-point messages sent by processes during the
execution of \alga\ could be as large as the expected time $\mathcal{O}\left(\sqrt{\frac{n}{(n-f)\log(n/(n-f))}}\right)$
multiplied by $\cO(n^2)$ messages sent each round by the communication
procedure \prop, which in the rough calculation could yield $\mathcal{O}\left(n^2\sqrt{\frac{n}{(n-f)\log(n/(n-f))}}\right)$
of total point-to-point messages sent.
Estimating more carefully, each round of \prop\ outside the \lcon\ execution
costs $\cO(n^2)$ point-to-point messages, and each message
exchange inside the \lcon\ costs $\cO\left(\frac{n^2}{(n-f)^2}\right)$
with constant probability,
which gives a more precise upper bound 
$\mathcal{O}\left(n^2+\sqrt{\frac{n^5}{(n-f)^5\log(n/(n-f))}}\right)$
on the expected number of point-to-point messages.

In order to improve message complexity, we modify \alga\
in the following way -- we call the resulting protocol \algb. 
Each communication round is implemented by using a gossip
protocol from~\cite{CK-DISC06}, which is a deterministic protocol guaranteeing
successful exchange of messages between processes which are non-faulty at
the end of this protocol in $O(\log^3 n)$ rounds using $O(m\log^4 n)$ 
point-to-point messages, where $m\le n$ is the known upper bound on the number of participating processes.
This gives $\cO(n\log^4 n)$ point-to-point messages per execution of \prop\ outside of \lcon\ procedure, and 
$\cO\left(\frac{n\log^4 n}{n-f}\right)$
point-to-point messages, with constant probability, in each exchange during the execution
of \lcon.
Therefore, \algb\ accomplishes consensus in $\mathcal{O}\left(\sqrt{\frac{n\log^5 n}{(n-f)}}\right)$ expected rounds
using $\mathcal{O}\left(\left(\frac{n}{n-f}\right)^{3/2} \log^{7/2} n + n\log^4 n\right)$ point-to-point messages in expectation. In particular,
for large range of crashes $f=n-\Omega(\sqrt{n})$, 
\algb\ is nearly optimal (i.e., with respect to a polylogarithm of $n$)
in terms of both expected time and message complexity.
For a discussion about lower bounds see Section \ref{sec4}.

\subsection{Lower bound}
\label{sec4}

\begin{theorem}
	\label{t:lower}
	The expected number of rounds of any consensus protocol
	running against a \wadapt\ or a \noadapt\ adversary causing up to $f$
	crashes is $\Omega\left(\sqrt{\frac{n}{(n-f)\log(n/(n-f))}}\right)$.
\end{theorem}

\begin{proof}
Consider a consensus algorithm.
We partition $n$ processes into $m=\frac{n}{n-f}$ groups of size $n-f$
and treat each group as a distributed entity similar as in the proof of the impossibility result for consensus against at least third part of Byzantine processes, c.f.,~\cite{AtWe}. 

We assume that entities are described by a tuple, where each position in the tuple corresponds to a certain processor and the elements of the tuple 
describe processors' states. Furthermore, we assume that communication between entities is by sending concatenated messages of all the processes gathered in a particular round.

Bar-Joseph and Ben-Or showed in~\cite{BB}
that the expected number of rounds of any consensus protocol
running against a \sadapt\ adversary causing up to $m/2$
crashes on $m$ processes is $\Omega\left(\sqrt{m/\log m}\right)$. 
We call this adversary the Bar-Joseph and Ben-Or adversary, for reference.
It automatically implies, by Markov inequality, 
that with probability 
at least $3/4$ an execution generated by the 
Bar-Joseph and Ben-Or adversary 
has length $\Omega\left(\sqrt{m/\log m}\right)$.
We apply this to the above mentioned system of $m$ groups.
This leads us to a class of executions that we denote $\mathcal{A}_{s}$, generated by the \sadapt\ adversary that simulates all the executions of given algorithms until half of the entities remain operational. 
We consider a class $\mathcal{A}_{s}$ as an execution tree that we define as follows.

Let us assume that the execution tree of an algorithm is a tree where each path from the root to the leaves represents the number of rounds a particular execution takes. What is more, nodes are labeled with random bits and the status of entities (which of them are operational) at a particular round of the execution.
Note that possible system configurations at round $i$ are represented
as nodes at level $i$ of the tree.
This may lead to an enormously huge tree, with each path/node having
its probability of occurrence; yet we assume that the adversary has sufficient computational power to build such a structure.

We call such an execution \emph{long}
if it has length $\Omega\left(\sqrt{m/\log m}\right)$.
Recall that the probability
of an event that an execution against the Bar-Joseph and Ben-Or adversary
is in class $\mathcal{A}_{s}$ and is long is at least $3/4$.

The tree is built along paths (corresponding to specific executions) until half of the entities remain operational, hence there is an entity (i.e., group) $ E $ that remained operational in at least half of these executions
(or more precisely, in executions which occur with probability at
least $1/2$), again by using a pigeonhole type of arguments with respect to the remaining groups and executions.
It follows that with probability at least $3/4-1/2=1/4$
a long execution in $\mathcal{A}_{s}$ contains group $E$.

Let us now consider a \wadapt\ adversary that chooses $ E $ as the group of $n-f$ fault-resistant processes, before the execution of the algorithm.
We consider a class of executions $\mathcal{A}_{w}$ generated by the \wadapt\ adversary of Bar-Joseph and Ben-Or, by simulating the same executions (until possible) as the \sadapt\ one considered above. Note that such adversary crashes whole groups.
With probability at least $1/4$ there is an execution containing $E$ (therefore, the \wadapt\ adversary does not stop it before the \sadapt\ does) and of length $\Omega\left(\sqrt{\frac{n}{(n-f)\log(n/(n-f))}}\right)$.
This implies the expected length of an execution caused
by the \wadapt\ adversary not failing group $E$
to be at least 
$1/4\cdot \Omega\left(\sqrt{\frac{n}{(n-f)\log(n/(n-f))}}\right)
=\Omega\left(\sqrt{\frac{n}{(n-f)\log(n/(n-f))}}\right)$.
%

In order to extend this bound into the model with \noadapt\ 
adversary, we recall that in the proof of the lower bound
in~\cite{BB}, in each round a suitable set of crashes
could be done with probability $1-O(1/\sqrt{m})$. Therefore,
by a union bound over $\Theta(\sqrt{m/\log m})$ rounds, there
is a path in the tree used in the first part of this proof
which leads to $\Theta(\sqrt{m/\log m})$ rounds without consensus,
with probability at least $1-1/\log m$.
The adversary selects this path, corresponding to the schedule of crashes of subsequent groups on the path in consecutive rounds,
prior the execution. This assures that in the course of the execution consensus is not achieved within the first $\Theta(\sqrt{m/\log m})$) rounds, with probability at least
$1-1/\log m$, thus also in terms of the expected number of rounds.
\end{proof}

Note that the lower bound on the expected number of point-to-point
messages is $\Omega(n)$ even for weaker non-adaptive adversaries, as each process has to receive a message to guarantee agreement in case of potential crashes, c.f.,~\cite{GK-SODA-10} for references.

\section{\chainadapt\ and \kthick\ adversaries}

In this section we present \algc\ - a modification of \alga\ specifically tailored to run against the \chainadapt\ adversary. 
\dk{Then we also} show that it is capable of running against a \kthick\ adversary. We begin with a basic description of the algorithm.

\subsection{\algc}
%
%

\begin{algorithm}
{\alga\ with \elect\ substituted by \electc\;}

\caption{\algc, \dk{pseudocode} for process~$v$}
\label{algorithm4}
\end{algorithm}

\begin{algorithm}
{initialize variable $ n^* $\;}
{$n^* := \countstations $\;}
{$ i = \lfloor n/n^* \rfloor $\;}
{$\texttt{coin} := \frac{k}{2^{i-1} n^{*}}$\;}
{initialize list $ \texttt{LEADERS} $ to an empty list\;}
{
{toss a coin with the probability $ \texttt{coin} $ of heads to come up\;}
{\If{heads came up in the previous step}{
    \prop(``$ v $'') to all other processes\;
    add $ v $ to list $\texttt{LEADERS}$\;}
}
}
{fill in list $ \texttt{LEADERS}$ with elected leaders' identifiers from received messages\;}

return $\texttt{LEADERS}$\;
\caption{\electc, \dk{pseudocode} for process~$v$}
\label{algorithm5}
\end{algorithm}

\begin{algorithm}
{\prop(``$v $'') to all other processes\;}
{return \texttt{the number of ID's heard} \;}

\caption{\countstations, \dk{pseudocode} for process~$v$}
\label{algorithm6}
\end{algorithm}


The algorithm begins with electing a number of leaders in \electc. However, as the adversary models its pattern of crashes into $ k $ disjoint chains then we would like to elect approximately $ k $ leaders. 

It may happen that the adversary significantly reduces the number of processes and hence the leader election procedure is unsuccessful in electing an appropriate number of leaders. That is why we adjust the probability of success by approximating the size of the network before electing leaders. If the initial number of processes was $ n $ and the drop in the number of processes after estimating the size of the network was \textit{not significant (less than half the number of the approximation)} then we expect to elect $ \Theta(k) $ leaders.

Otherwise, if the number of processes was reduced by more than half, the probability of success is changed and the expected number of elected leaders is reduced. This helps to shorten executions of \lcon\, because a smaller number of leaders executes the protocol faster.
In general if there are $ \frac{n}{2^{i}} $ processes, we expect to elect $ \Theta\left(\frac{k}{2^{i}}\right) $ leaders.

Elected leaders are expected to be placed uniformly in the adversary's order of crashes. If we look at a particular leader~$ v $, then he will be present in some chain $ k_i $. What is more, his position within this chain is expected to be in the middle of~$ k_i $.

Leaders execute the small consensus protocol \lcon. If they reach consensus, then they communicate this fact twice to the rest of the system. Hence, if the adversary wants to prolong the execution, then it must crash all leaders. Otherwise, the whole system would reach consensus and end the protocol.

If leaders are placed uniformly in the adversary's order, then the adversary must preserve the pattern of crashes that it declared at first. In what follows, if there is a leader $ v $ that is placed in the middle of chain $ k_i $, then half of the processes preceding $ v $ must also be crashed.

When the whole set of leaders is crashed then another group is elected and the process continues until the adversary spends all its possibilities of failing processes.

\begin{theorem} \label{theorem2}
The following hold \algc\ against the \chainadapt\ adversary:

\noindent
\textbf{a)} \algc\ reaches consensus in the expected number of rounds equal $ \mathcal{O}(\timelc(k)\log(n/k))$, satisfying termination, agreement and validity.

\noindent
\textbf{b)} Expected work per process of \algc\ is $ \frac{k}{n} \log\left(\frac{n}{k}\right) \sum_{i = 0}^{\log k} \mathcal{O}\left(\timelc\left(\frac{k}{2^i}\right)\right) $.
 
 \noindent
\textbf{c)} \algc\ reaches consensus in the number of rounds equal $ \log\left(\frac{n}{k}\right)\sum_{i = 0}^{\log k} \mathcal{O}\left(\timelc \left(\frac{k}{2^i}\right)\right) $ and with work per process $ \frac{k}{n} \log\left(\frac{n}{k}\right) \sum_{i = 0}^{\log k} \mathcal{O}\left(\timelc\left(\frac{k}{2^i}\right)\right) $ at least with probability $ 1 - e^{-\frac{k^{1/4}}{24} + \log\log k^2} $.
 
\end{theorem}

\begin{corollary} \label{cor2}
 Instantiating \lcon\ with \synran\ from \cite{BB} results in:
 
 \noindent
  \textbf{a)} $ \mathcal{O}\left(\sqrt{\frac{k}{\log k}}\log(n/k)\right) $ expected number of rounds to reach consensus by \algc,
  
  \noindent
  \textbf{b)} expected work per process equal $ \frac{k}{n} \log\left(\frac{n}{k}\right) \mathcal{O}\left(\sqrt{\frac{k}{\log k}}\log(n/k)\right) $,
  
  \noindent
  \textbf{c)} $ \mathcal{O}\left(\sqrt{\frac{k}{\log k}}\log(n/k)\right) $ rounds to reach consensus by \algc\ with probability $ 1 - e^{-\frac{k^{1/4}}{24} + \log\log k^2} $. 
\end{corollary}

\subsection{\algc\ against the adversary limited by an arbitrary partial order}

Let us consider the adversary that is limited by an \textbf{arbitrary} partial order $P=(P,\succ)$. Two elements in this partially ordered set 
are \textit{incomparable} if neither $x \succ y $ nor  $y\succ x$ hold. Translating this into our model, the adversary may 
crash incomparable elements in any sequence during the execution of the algorithm.

\begin{theorem} \label{theorem3}
\algc\ reaches consensus in expected $ \mathcal{O}(\timelc(k)\log(n/k))$
number of rounds, against the \kthick\ adversary, satisfying termination, agreement and validity.
\end{theorem}

\subsection{Lower bound}

\begin{theorem}
\label{t:lb-chain}
	For any reliable randomized algorithm solving consensus in a message-passing model
	and any integer $0<k\le f$, 
	there is a 
	\kthick\ adversary that
	can force the algorithm to run in 
	$\Omega(\sqrt{k/\log k})$ expected number
	of rounds.
\end{theorem}

\begin{proof}[Sketch]
The adversary takes the order consisting of $n/k$ independent
chains, and follows the analysis for \sadapt\ adversaries in~\cite{BB}, with the following
modification: if the analysis enforces crashing a process $v$,
the adversary crashes the whole chain to which $v$ belongs.
Since there are $k$ chains and $k-1$ possible chain-crashes,
by replacing $n$ by $k$ in the formula obtained in~\cite{BB} we
get $\Omega(\sqrt{k/\log k})$ expected number
of rounds.
\end{proof}

\section{Conclusions and open problems}

In this work we showed time and message efficient randomized consensus
against the \wadapt,
\noadapt\ 
and \ordered\ adversaries generating crashes. 
We proved that \dk{all these classes of \constrained} adaptive adversaries are weaker than the
\sadapt\ one.
Our results also extend the state-of-the-art of the study
of popular \noadapt\ adversaries.

\bibliographystyle{plain}
\bibliography{consensus2}

\begin{appendix}
 \section*{Appendix} \label{append}

\section{Additional related work and open problems}


\paragraph{Additional related work.}

Fisher, Lynch and Paterson~\cite{FLP} showed that for the message passing model consensus cannot be solved  deterministically in the {\em asynchronous settings}, even if only one process may crash.
Loui and Abu-Amara~\cite{LA} showed a corresponding result for shared memory.
These impossibility results can be circumvented when randomization is used and the consensus termination condition does not hold with probability~$1$.

Bracha and Toueg~\cite{BT} observed that it is impossible to reach consensus by a randomized algorithm in the asynchronous model with crashes if the majority of processes are allowed to crash.
Ben-Or~\cite{Ben} gave the first randomized algorithm solving consensus in the asynchronous message passing model under the assumption that the majority of processes are non-faulty.
For the strong adversary Abrahamson~\cite{Abr} gave the first randomized solution in the shared memory model, which had exponential work, and Aspnes and Herlihy~\cite{AH} gave the first polynomial-work algorithm for the same adversary.
Aspnes and Waarts~\cite{AsWa}, Bracha and Rachman~\cite{BR}, and Aspnes, Attiya and Censor~\cite{AAC} gave solutions with $\cO(n\,\polylog n)$ work per process.
Aspnes~\cite{Asp-JACM} showed that this was best possible up to a poly-logarithmic factor by showing that any algorithm for the strong adversary requires $\Omega(f^2/\log^2 f)$ work by proving that these many coin flips were needed.
Attiya and Censor~\cite{AtCe} resolved the work complexity of randomized consensus in shared memory against the strong adversary by showing that $\Theta(n^2)$ total work is both sufficient and necessary, building on the work of Aspens and Herlihy~\cite{AH}, Bar-Joseph and Ben-Or~\cite{BB}, and Moses and Rajsbaum~\cite{MR}.
Research on weak adversaries~\cite{Aum,AuBe,Cha}, that may observe coin flips and the content of shared registers, resulted in scalable algorithms of $\cO(n\,\polylog n)$ total work and $\cO(\polylog n)$ work per process.

 
Deterministic consensus solutions operating in time $\cO(f)$ with sub-quadratic message complexities usually rely on gossiping procedures, as given for instance in~\cite{CK-JCSS06,CK-DISC06,GGK}.
All such algorithms  send $\Omega(n^2)$ bits in messages in the worst case, by the nature of gossiping.
Solutions given in~\cite{CK-DISC06,GGK} are \emph{early stopping} i.e. they terminate in~$\cO(f)$ rounds where $f\le n$ is the number of failures that actually occur in a given execution.
Georgiou, Gilbert, Guerraoui and Kowalski~\cite{GGGK} gave a randomized solution for asynchronous message passing systems with processes prone to crashes which sends $o(n^2)$ messages, but their algorithm has $\omega(n^2)$ bit complexity.


{\em Deterministic Byzantine} consensus requires $\Omega(f n)$  message complexity, which can be improved with the use of randomization.
A Byzantine consensus algorithm achieves \emph{consensus with loss~$h(f)$} if at least $n-h(f)$ non-faulty processes eventually decide on a common value when up to~$f$ processes are corrupted.
An \emph{almost-everywhere consensus} means consensus with loss~$h(f)$ for some function $h(f)\le c f$, where~$c<1$ is a constant.
The first randomized solutions for almost-everywhere Byzantine consensus and leader election scalable with respect to bit communication complexity were given by King, Saia, Sanwalani and Vee~\cite{KSSV-FOCS06}.
A tradeoff between the time required to solve an almost everywhere Byzantine consensus by a randomized scalable solution and the magnitude of loss was proved by Holtby, Kapron and King~\cite{HKK}.
They showed that if $f$ is a constant fraction of~$n$ then a scalable solution requires $\Omega(n^{1/3})$ rounds.

The consensus problem has been recently considered against different adversarial scenarios. Robinson, Scheideler and Setzer \cite{RSS} considered the synchronous consensus problem under 
a late $\epsilon$-bounded adaptive adversary, whose observation of the system is delayed by one round and can block up to $ \epsilon n $ nodes in the sense that they cannot receive and send messages 
in a particular round.

\paragraph{Open problems.}
Three main open directions emerge from this work.
One is to improve message complexity for the whole range of the number
of crashes $f$,
especially those close to $n$ by less than $\sqrt{n}$.
Another open direction could pursue a study of complexities of other
important distributed problems and settings against 
\wadapt\ and \ordered\ adversaries, \dk{which
	are more realistic than the \sadapt\ one and 
	more general than the 
	\noadapt\ 
	one, commonly used in the literature.}
Finally, there is a scope of proposing and studying other intermediate types of adversaries, including
further study of recently proposed 
delayed adversaries~\cite{KKM} and
adversaries tailored for dynamic distributed and parallel computing~\cite{KM-ICALP2018}.

\section{\alga\ flow diagram}

\begin{figure*}[htb]
  \begin{center}
    \includegraphics[scale=0.63]{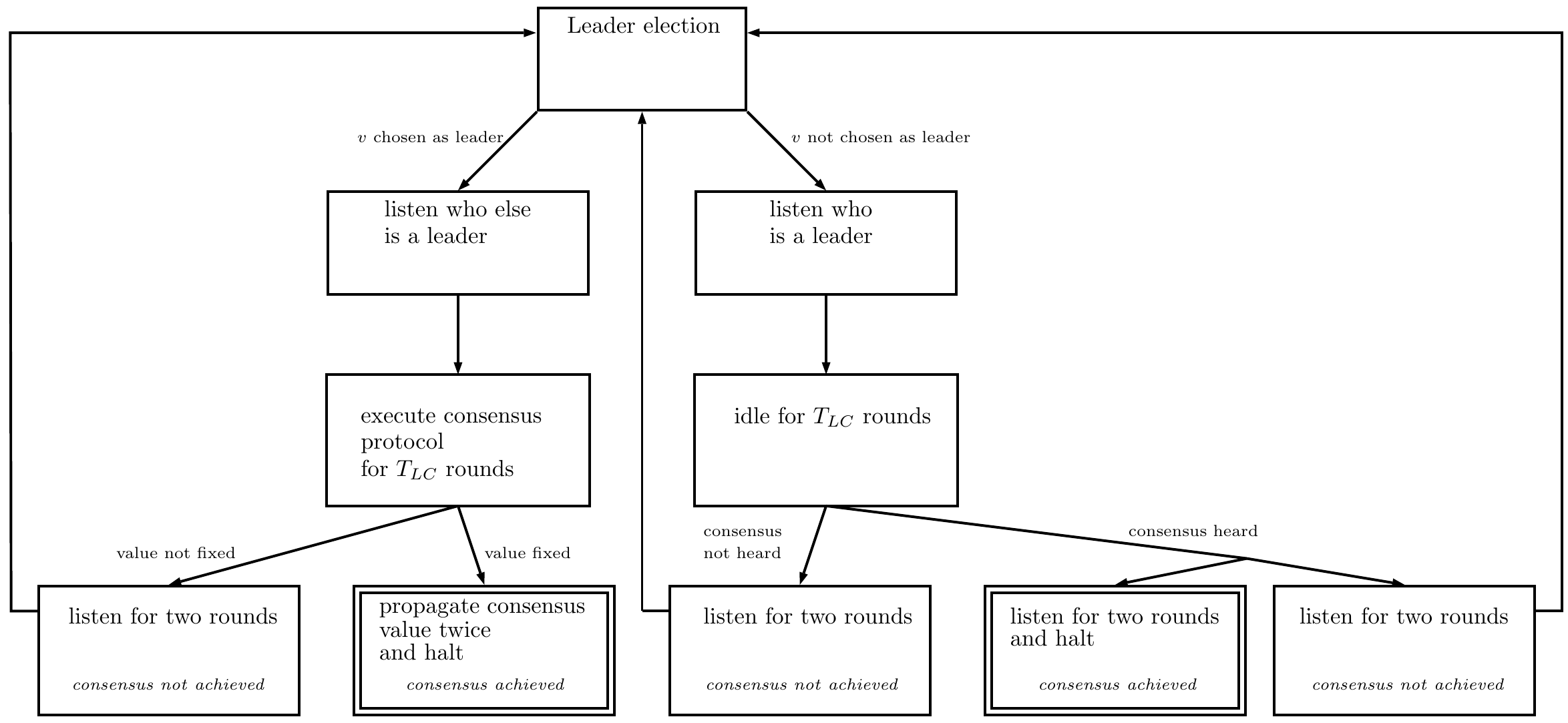}
    \caption{\alga\ flow diagram for process $ v $.}
    \label{fig1}
  \end{center}
\end{figure*}

\section{Proof of Lemma \ref{lemma1}}

\begin{proof}

Let $ X $ be a random variable such that $ X = X_{1} + \cdots + X_{n}, $ where $ X_{1}, \cdots, $ $X_{n} $ are Poisson trials and
 
$
X_{i} = \left\{ \begin{array}{ll}
1 & \textrm{if process i chosen a leader,}\\
0 & \textrm{otherwise.}
\end{array} \right 
.$

We know that $ \mu = \mathbb{E}X = \mathbb{E}X_{1} + \cdots + \mathbb{E}X_{n} = \frac{n}{n-f}. $
To estimate the probability that the number of leaders lies within the expected interval we will use the following Chernoff's inequality:
$ \mathbb{P}[X \geq R] \leq 2^{-R} $, where $ R \geq 6\mu. $

Let $ R = 6\frac{n}{n-f} $. Then
$ \mathbb{P}\left[X \geq 6\frac{n}{n-f}\right] \leq 2^{-6\frac{n}{n-f}} \leq \frac{1}{2^{6}} \leq \frac{1}{10}. $
In what follows with probability $ \frac{9}{10} $ the number of leaders is less than $ 6\frac{n}{n-f} $. 
\end{proof}

\section{Proof of Lemma \ref{lemma2}}

\begin{proof}
 There are $ n $ processes out of which $ f $ are prone to crashes. Hence the system contains $ n - f $ fault-resistant processes. The probability that a fault-resistant one will respond
 in the election procedure is $ \frac{1}{n-f} $. We estimate the probability of electing at least one fault-resistant leader by a complementary event that none of the fault-resistant processes responded during
 the election procedure:
 $ \left(1 - \left(1 - \frac{1}{n-f}\right)^{n-f}\right) \geq (1-e^{-1}) \geq \frac{6}{10}.$
 
\end{proof}

\section{Proof of Theorem \ref{theorem1} b) and c)}

\begin{proof}

\noindent
   \textbf{b)} In each iteration of \alga\ a process either becomes a leader or stays idle, waiting to hear and adopt the consensus value decided by leaders. This leads to an observation that the actual work is done by processes taking part in \lcon. We know that $ \frac{5}{2} $ is the expected number of iterations of \alga\ and each process becomes a leader in each iteration with probability $ \frac{1}{n-f} $. The expected number of rounds in each iteration is $ \mathcal{O}(\timelc(\frac{6n}{n-f}))$, hence by applying Wald's equation (\cite{MU}, Theorem 12.3) we have that $ \frac{1}{n-f} \; \mathcal{O}(\timelc(\frac{6n}{n-f}))$ is the expected work per process.
   
   \noindent
   \textbf{c)} We know that $ \mathbb{P}(Z) \geq \frac{2}{5} $. Let $ \epsilon > 0 $ be fixed and $ X \sim Geo\left(\frac{2}{5}\right) $. We have that
   $ \mathbb{P}(X \geq i) = \left(1 - \frac{2}{5}\right)^{i-1} $. 
   
   Taking $ i = \frac{5}{2}\log\left(\frac{1}{\epsilon}\right) + 1 $, we have that 
   $ \mathbb{P}\left(X \geq \frac{5}{2}\log\left(\frac{1}{\epsilon}\right) + 1\right) \leq e^{-\log\left(\frac{1}{\epsilon}\right)} = \epsilon $. Thus, $ \mathbb{P}\left(X < \frac{5}{2}\log\left(\frac{1}{\epsilon}\right) + 1\right) > 1 - \epsilon $.
   
   \noindent
   Applying this to the results from points a) and b) of the proof gives the desired result.   

\end{proof}

\section{Proof of Theorem \ref{theorem2}}

In order to prove Theorem \ref{theorem2} we need the following Lemma

\begin{lemma}
 Let $ L $ denote the number of leaders elected by \electc, $ n^* $ be the number of operational processes and $ k $ denote the number of chains in the adversary's order and $ l $ denote the number of leaders that the protocol would like to elect, where $ l = \frac{k}{2^i} $ for $ 0 \leq i \leq \log k $. Then $ l/4 < L < 3l/4 $ at least with probability $ 1 - 2e^{-l/24} $.
 \label{lemma3}
\end{lemma}

\begin{proof}

Let $ X $ be a random variable such that $ X = X_{1} + \cdots + X_{n^*}, $ where $ X_{1}, \cdots, $ $X_{n^*} $ are Poisson trials and
 
$
X_{j} = \left\{ \begin{array}{ll}
1 & \textrm{if process i chosen a leader,}\\
0 & \textrm{otherwise.}
\end{array} \right 
.$

Since the number of operational processes is $ n^* $ and the probability of success is set to at most $ \frac{l}{2n^*} $ in the \electc\ procedure, we know that $ \mu \leq \mathbb{E}X = \mathbb{E}X_{1} + \cdots + \mathbb{E}X_{n^*} = \frac{l}{2}. $
To estimate the probability that the number of leaders lies within the expected interval we use the following Chernoff's inequalities:
\begin{enumerate}
 \item $ \mathbb{P}[X \geq (1 + \delta)\mu] \leq e^{-\mu\delta^{2}/3}, $ for $ 0 < \delta \leq 1 $.
 \item $ \mathbb{P}[X \leq (1 - \epsilon)\mu] \leq e^{-\mu\epsilon^{2}/2}, $ for $ 0 < \epsilon < 1 $.
\end{enumerate}
\noindent
Taking $ \delta = \epsilon = \frac{1}{2} $, we have that
$ \mathbb{P}\left[X \geq \frac{3l}{4}\right] \leq e^{-l/24} $
and 
$ \mathbb{P}\left[X \leq \frac{l}{4}\right] \leq e^{-l/16} $.
Hence $ \mathbb{P}\left[X < \frac{3l}{4}\right] > 1 - e^{-l/24} $
and 
$ \mathbb{P}\left[X > \frac{l}{4}\right] > 1 - e^{-l/16} $.
In what follows $ l/4 < L < 3l/4 $ at least with probability $ 1 - 2e^{-l/24} $.

\end{proof}

\noindent 
We now proceed with the proof of Theorem \ref{theorem2}

\begin{proof}
\noindent
\textbf{a)}
 \algc\ is structured as a repeat loop at the beginning of which a leader election procedure is executed. \electc\ approximates the size of the network in the first place by a very simple procedure \countstations\ in order to adjust the coin values appropriately. 
 
 After the number of processes is counted it may happen that the adversary crashes some of the processes, so that their number drops significantly i.e. by more than half the number returned by \countstations. We will, however, firstly analyze the opposite situation when the drop is not significant to explain the basic situation and calculate the resulting time for the algoritm to succeed. Then we will proceed to the analysis of the situation when the adversary reduces the number of operational processes significantly and calculate a union bound over all such cases.
 
%


 According to Lemma \ref{lemma3} \electc\ returns $ \Theta(k) $ leaders. In what follows the expected number of rounds required to reach small consensus by \lcon\ is $ \mathcal{O}(\timelc(k)) $.
 
 Note that the adversary has to crash all the leaders at some point of the \lcon\ protocol to prolong the execution. Otherwise, there would be a processe that would reach small consensus and inform all other processes about the consensus value.
 
 According to the adversary's partial order there are initially $ k $ chains, where chain $ j $ has length $ l_{j} $. If a leader is elected then it belongs to one of the chains. We show that it is expected that the chosen leader is placed in the middle of that chain.

 Let $ X $ be a random variable such that $ X_{j} = i $ where $ i $ represents the position of the leader in chain $ j $. We have that
 $ \mathbb{E}X_{j} = \sum_{i=1}^{l_{j}} \frac{i}{l_{j}} = \frac{1}{l_{j}} \frac{(1+l_{j})}{2}l_{j} = \frac{(1+l_{j})}{2}. $

 If the leader is crashed, this implies that half of the processes forming the chain are also crashed.
 If at some other points of time, the leaders will also be elected from the same chain, then by simple induction we may conclude that this chain is expected to be all crashed after $ \mathcal{O}(\log(n/k)) $ steps in average.
 
 In what follows if there are $ k $ chains and we expect to elect $ \Theta(k) $ leaders, then after at most $ \mathcal{O}(\log(n/k)) $ steps this process ends, as the adversary runs out of all the possibilities of crashing. Assuming that the adversary crashes the last leader just before it is to communicate the consensus value we may estimate that each such step needs $ \mathcal{O}(\timelc(k)) $ time for running \lcon. Hence, the expected time for \algc\ to reach consensus is $ \mathcal{O}(\timelc(k)\log(n/k))$.
 
 Let us now proceed to the situation when the adversary crashes a significant number of processes (i.e. more than half of them). In such case procedure \electc\ changes the probability of success of electing the leaders as well as reduces the expected number of leaders that should be elected in the next iteration of the algorithm.
 
 Once again, according to Lemma \ref{lemma3}, we expect procedure \electc\ to return $ \Theta(k/2^i) $ leaders for some $ 0 \leq i \leq \log k $, since the number of such drops is bounded by $ \log k $. Calculating the time required to reach consensus over all cases gives us $ \sum_{i = 0}^{\log k} \mathcal{O}\left(\timelc \left(\frac{k}{2^i}\right) \log\left(\frac{n}{2^i k}\right)\right) = \log\left(\frac{n}{k}\right)\sum_{i = 0}^{\log k} \mathcal{O}\left(\timelc \left(\frac{k}{2^i}\right)\right) $. 

 Altogether we have that \algc\ requires $ \log\left(\frac{n}{k}\right)\sum_{i = 0}^{\log k} \mathcal{O}\left(\timelc \left(\frac{k}{2^i}\right)\right) $ expected number of rounds to reach consensus, where termination, validity and agreement are satisfied because of the same arguments as in Theorem \ref{theorem1}.
 
  \noindent
 \textbf{b)} 
 Similarly as in the proof of Theorem \ref{theorem1}, if there are $ n^* $ operational processes, and there were no significant drops in this number, we know that $ \log\left(\frac{n}{k}\right) $ is the expected number of iterations of \algc\ and each process becomes a leader in each iteration with probability $ \frac{k}{n^*} $. The expected number of rounds in each iteration is $ \mathcal{O}(\timelc(k))$, hence by applying Wald's equation (\cite{MU}, Theorem 12.3) we have that $ \frac{k}{n^*} \; \mathcal{O}(\timelc(k))$ is the expected work per process.
 
 However assuming that significant drops may happen while \lcon\ procedure, we need to calculate a union bound over all the cases, which for each $ 0 \leq i \leq \log k $ gives us $ \log\left(\frac{n}{2^{i}k}\right) $ expected number of iterations of \algc\ and the probability that each process becomes a leader equal $ \frac{k}{2^{i}\frac{n}{2^i}} $. The expected number of rounds in each iteration for a particular $ i $ is $ \mathcal{O}\left(\timelc\left(\frac{k}{2^i}\right)\right) $. Summing this altogether we have that the expected work per process is equal $ \sum_{i = 0}^{\log k} \frac{k}{n} \; \mathcal{O}\left(\timelc\left(\frac{k}{2^i}\right)\right) \log\left(\frac{n}{2^{i}k}\right) = \frac{k}{n} \log\left(\frac{n}{k}\right) \sum_{i = 0}^{\log k} \mathcal{O}\left(\timelc\left(\frac{k}{2^i}\right)\right)$.
 
 \noindent
 \textbf{c)} The success of the algorithm depends on successful leader election through all the cases of significant drops in the number of operational processes.
 According to Lemma \ref{lemma3} for each $ 0 \leq i \leq \log k $ the probability that \lcon\ is successful in returning $ \Theta\left(\frac{k}{2^i}\right) $ leaders is $ 1 - 2e^{-\frac{k}{2^i}\frac{1}{24}} $. We want to calculate the probability that \lcon\ is successful over all cases. Let $ A_{i} $ be an event that is was successful in the $ i $'th case. We have that
 $ \mathbb{P}\left(\bigcap_{i = 1}^{\log k}A_{i}\right) = \bigcap_{i = 1}^{\log k} \left(1 - 2e^{-\frac{k}{2^i}\frac{1}{24}}\right) 
 \geq \bigcap_{i = 1}^{\log k} \left(1 - 2e^{-\frac{k}{2^{\log k}}\frac{1}{24}}\right). $
 Since $ 2^{\log k} = e^{\log 2 \log k} = e^{\log k^{\log 2}} = k^{\log 2}$ and $ \frac{1}{2} < \log 2 < 1 $, thus
 $ \bigcap_{i = 1}^{\log k} \left(1 - 2e^{-\frac{k}{2^{\log k}}\frac{1}{24}}\right) \geq \bigcap_{i = 1}^{\log k} \left(1 - 2e^{- \frac{k^{1/4}}{24}}\right)
 = \left(1 - 2e^{-\frac{k^{1/4}}{24}}\right)^{\log k}. $
 Using the Bernoulli inequality we continue having
 $ \left(1 - 2e^{-\frac{k^{1/4}}{24}}\right)^{\log k} \geq 1 - 2\log k \; e^{-\frac{k^{1/4}}{24}} 
 = 1 - e^{-\frac{k^{1/4}}{24} + \log\log k^2}.$
Applying the result above to points a) and b) finishes the proof.

\end{proof}

Since, similarly as in the analysis of \alga\, we would like to instantiate procedure \lcon\ with algorithm \synran\ from \cite{BB}, we introduce the following fact, allowing us to sum the time complexity appropriately in Corrollary \ref{cor2}.

\begin{fact} \label{fact1}
Let $ n, k \in \mathbb{N} $, where $ 0 \leq k \leq n $. Then
 $\sum_{m=0}^{\log k} \sqrt{\frac{k}{2^m\log\left(\frac{k}{2^m}\right)}} \log\left(\frac{n}{2^m k}\right) = c\; \sqrt{\frac{k}{\log k}} \log\left(\frac{n}{k}\right), $ for some constant $ c $.
\end{fact}

\begin{proof}
We have that
 $ \sum_{m=0}^{\log k} \sqrt{\frac{k}{2^m\log\left(\frac{k}{2^m}\right)}} \log\left(\frac{n}{2^m k}\right) \leq 
 \log\left(\frac{n}{k}\right) \sum_{m=0}^{\log k} \sqrt{\frac{k}{2^m\log\left(\frac{k}{2^m}\right)}}. $
 Since $ m \leq \log k $, then $ \log\left(\frac{k}{2^m}\right) \geq \log\left(\frac{k}{2^{\log k}}\right) = \log(k^{1 -\log 2}) = (1 - \log 2)\log k \geq \frac{1}{4} \log k, $
 because $ 2^{\log k} = e^{\log 2 \log k} = e^{\log k^{\log 2}} = k^{\log 2}$ and $ \frac{1}{2} < \log 2 < 1 $.
 Hence
$ \log\left(\frac{n}{k}\right) \sum_{m=0}^{\log k} \sqrt{\frac{k}{2^m\log\left(\frac{k}{2^m}\right)}}
 \leq 2 \log\left(\frac{n}{k}\right) \sum_{m=0}^{\log k} \sqrt{\frac{k}{2^m\log k}} 
 \leq 2 \sqrt{\frac{k}{\log k}}\log\left(\frac{n}{k}\right) \sum_{m=0}^{\log k} 2^{-m/2} 
 \leq c \; \sqrt{\frac{k}{\log k}}\log\left(\frac{n}{k}\right)
 $.

\end{proof}

\section{Proof of Theorem \ref{theorem3}}

\begin{proof}
We assume that crashes forced by the adversary are constrained by some partial order $P$. 
Let us recall the following lemma. 

\begin{lemma}\label{DWlemma}(Dilworth's theorem \cite{Dilworth})
In a finite partial order, the size of a maximum anti-chain is equal to the minimum
number of chains needed to cover all elements of the partial order.
\end{lemma}

The \kthick\ adversary is
constrained by an arbitrary order of thickness~$k$.
The adversary by choosing some $f$ processes to be crashed cannot increase the size of the maximal anti-chain. 
Thus using Lemma~\ref{DWlemma} we consider the coverage of the crash-prone processes by at most $k$ disjoint chains and any dependencies between chains' elements
create additional constraints to the adversary comparing to the \chainadapt\ one. Hence we fall into the case 
concluded in Theorem~\ref{theorem2} that completes the proof. 

\end{proof}

\end{appendix}

\end{document}